\documentclass{article}
\usepackage{amssymb,amsmath,amsthm,bbm,eucal}
\usepackage{tikz}
\usepackage{paralist}
\usepackage{xcolor}

\definecolor{reddish}{rgb}{1,.8,0.8}
\definecolor{blueish}{rgb}{0.8,.8,1}
\definecolor{greenish}{rgb}{.8,1,0.8}
\definecolor{yellowish}{rgb}{1,1,.20}

\definecolor{webred}{rgb}{0.5,0,0}
\definecolor{webblue}{rgb}{0,0,0.8}
\newtheorem{theorem}{\bfseries Theorem}
 \newtheorem{lemma}[theorem]{\bfseries Lemma}

 \newtheorem{assumption}{\bfseries Assumption}

 \newtheorem{definition}{\bfseries Definition}
 \newtheorem{prop}[theorem]{\bfseries Proposition}

\newcounter{rem}
\newcommand{\Alg}{{\bf Alg\/} }

\newcommand{\num}[1]{\relax\ifmmode \mathbb #1\else $\mathbb #1$\fi}
\newcommand{\nnnum}[1]{\relax\ifmmode 
  {\mathbb #1}_{\geq 0} \else ${\mathbb #1}_{\geq 0}$
  \fi}
\newcommand{\npnum}[1]{\relax\ifmmode 
  {\mathbb #1}_{\leq 0} \else ${\mathbb #1}_{\leq 0}$
  \fi}
\newcommand{\pnum}[1]{\relax\ifmmode 
  {\mathbb #1}_{> 0} \else ${\mathbb #1}_{> 0}$
  \fi}
\newcommand{\nnum}[1]{\relax\ifmmode 
  {\mathbb #1}_{< 0} \else ${\mathbb #1}_{< 0}$
  \fi}
\newcommand{\plnum}[1]{\relax\ifmmode 
  {\mathbb #1}_{+} \else ${\mathbb #1}_{+}$
  \fi}
\newcommand{\nenum}[1]{\relax\ifmmode 
  {\mathbb #1}_{-} \else ${\mathbb #1}_{-}$
  \fi}

\newcommand{\reals}{{\num R}}                    
\newcommand{\nnreals}{{\nnnum R}}                    
\newcommand{\naturals}{{\num N}}                      
  
\newcommand{\N}{\mathcal{N}}
\newcommand{\E}{\mathcal{E}}
\newcommand{\T}{\mathbb{T}}
\newcommand{\exec}[1]{\relax\ifmmode {\sf Execs}_{#1} \else ${\sf Exec}_{#1}$\fi} 

\newcommand{\deq}{\mathrel{\stackrel{\scriptscriptstyle\Delta}{=}}}

\begin{document}
\title{Differentially Private Iterative Synchronous Consensus}
\author{
Zhenqi Huang \and Sayan Mitra \and Geir Dullerud \\
University of Illinois at Urbana-Champaign \\ 
\texttt{\{zhuang25,mitras,dullerud\}@illinois.edu }\\
} 
\maketitle

\begin{abstract}
The iterative consensus problem requires a set of processes or agents with different initial values, to interact and update their states to eventually converge to a common value. Protocols solving iterative consensus serve as building blocks in a  variety of systems where distributed coordination is required for load balancing, data aggregation, sensor fusion, filtering, clock synchronization and platooning of autonomous  vehicles. In this paper, we introduce the {\em private iterative consensus problem\/} where agents are required to converge while  protecting the privacy of their initial values from honest but curious adversaries.  Protecting the initial states, in many applications, suffice to protect all subsequent states of the individual participants. 

First, we adapt the notion of differential privacy in this setting of iterative computation. Next, we present a server-based  and a completely distributed randomized mechanism for solving private iterative consensus with adversaries who can observe the messages as well as the internal states of the server and a subset of the clients. Finally, we establish the tradeoff between privacy and the accuracy of the proposed randomized mechanism.
\end{abstract}


\section{Introduction}
\label{sec:intro}



This paper addresses the problem of reaching agreement in a group iteratively while preserving individual's privacy. 
The setup consists of $N$ agents, each with some initial information modeled as the valuation of a variable. 
The problem requires the agents to interact with each other and update their internal states, so that eventually they all converge to a common decision or value. This agreement to a common decision can then be used for coordinating the actions of the participating agents. Indeed, this {\em iterative consensus\/} has been used as a building block for designing a variety of distributed coordination protocols for load balancing~\cite{Cy89,Xiao:2007:DAC:1222667.1222952},  filtering and sensor fusion~\cite{Olfati-saber06distributedkalman,Xiao:2005:SRD:1147685.1147698}, clock synchronization, and flocking~\cite{Blondel,Saber.Murray2003Flockingwithobstacle,JLM03:TAC,TJG:TAC07,JM:JNSA2011}, to name a few. 

A natural, synchronous, and widely studied consensus mechanism involves, at each round, for every agent to update its state as a  weighted average of its own value with values of its neighboring agents. This update rule can be expressed as $x(t+1) = P x(t)$, where $x(t)$ is the vector of agent values and $P$ is a symmetric $N \times N$ matrix with $P_{ij}$ defining the communication weight between agents $i$ and $j$. It turns out that this class of consensus mechanisms\footnote{We refrain from calling these mechanisms algorithms because they are designed to converge and not to terminate.}  
converge to the average of the initial values of the agents and a measure of the speed of convergence is given by the second largest eigenvalue in absolute value of the matrix $P$. More general necessary and sufficient conditions for achieving consensus with  synchronous mechanisms, including cases where the matrix $P$ is time-varying, have been studied in~\cite{Tsitsiklis1984ProblemsinDecentralized,MurrayOlfati} (see the book for a complete overview~\cite{Magnusbook2010}). Sufficient conditions for achieving consensus with message delays and losses has been developed in~\cite{Tsitsiklis:1986,Bertsekas:1989:PDC:59912} and more recently, a theorem prover-based verification framework for these mechanisms has been presented in~\cite{MC:TPHOLS08,MC:FORMATS08}.
Furthermore, stochastic variants of the convergence mechanism under the presence of communication noises has been studied in~\cite{Xiao:2007:DAC:1222667.1222952,Minyi10}.

In this paper we study the {\em private consensus\/} problem which requires the agents to preserve the privacy of their initial values from an adversary who can see all the messages being exchanged, while also achieving convergence to the average of the initial values. The notion of privacy used in this paper is derived from the idea of differential privacy, first introduced in~\cite{DiffPri:Dwork06} (see~\cite{Dwork:2008:DPS:1791834.1791836} for a survey) in the context of ``one-shot'' computations on statistical databases. Roughly speaking, differential privacy ensures that the removal (or addition) of a single participant  from a database does not affect the output of any analysis {\em substantially\/}. It follows that an adversary looking at the output of any analysis cannot threaten to breach the privacy and security of individual participants. 

In~\cite{Dwork10}, the notion of differential privacy is expanded along two dimensions. First, it included  streaming and online computations in which the adversary can look at the entire sequence of outputs from the analysis algorithm. Secondly, it allowed the adversary to look at the internal state of the algorithm (Pan privacy) in addition to the communication messages.

This work is motivated by {\em closed-loop\/} applications where the output of the analysis is used as a feedback by the participating agents in updating their states. As a starting point in this investigation, we use a client-server setup for iterative consensus. The clients are the agents with private initial values. In each round, the clients send some information to the server based on their current state, the server updates its own state based on clients' information and sends feedback to the clients. Finally, the clients update their state according to some local control law based on the server's feedback. 
The clients require to converge, while their initial values should be protected from any honest but curious adversary with access to the messages (between  the clients and the server) as well as the server's internal state. We call this the {\em Synchronous Private Consensus (SPC)\/} problem. 

In many  distributed control systems, protected initial information imply protection of the current state. For example, consider a platoon of vehicles which require to move as a group with the same speed, while keeping their positions private. If the agents use a solution to the SPC problem for deciding on the common speed, then their initial velocities as well as their positions will be protected even if their initial positions and control laws are compromised.  

In Section~\ref{sec:pconsensus} we propose a randomized mechanism for solving the SPC problem. 
The key idea is to add a particular type of random noise to the clients' messages to  the server. 
Specifically, for a client with internal state $\theta(t)$ at round $t$, the message it sends to the server is $\theta(t) + \eta(t)$ where $\eta(t)$ is a random (real) number chosen according to a Laplace distribution with a parameter that decays geometrically with $t$. In contrast, the noise values added in~\cite{Dwork10} for implementing an approximate online counter are always chosen from the same Laplace distribution. 
The feedback $y(t)$ provided by the server is the mean of the noisy messages it receives.
And, the clients update their states by taking a linear combination of $y(t)$ and their earlier state.  
This weighted average is an example of a simple type of client dynamics. 

In Section~\ref{sec:distributed}, we generalize the client-server mechanism to a distributed setting
where the adversary can access the messages and the states of a subset of compromised clients. 
The mechanism guarantees differential privacy of the good clients and we 
derive a sufficient condition for convergence based on the communication pattern of the clients. 

As randomization is used for achieving privacy, this mechanism guarantees convergence to the average in a probabilistic sense: Given a probability $b$ and a radius $r$, we say that the mechanism is $(b,r)$-{\em accurate\/} if from any initial state, with probability $(1-b)$ the system converges to a value within $r$ distance of the average.  In Section~\ref{sec:discussion}, we discuss the tradeoff between privacy and accuracy. There are two parameters in the definition of the mechanism which can be chosen to get different levels of privacy and accuracy. If these parameters are tuned to obtain $\epsilon$-differential privacy, then we show that the accuracy that can be achieved is $(b, O(\frac{1}{\epsilon\sqrt{b N}})$. That is, the accuracy radius depends inversely on privacy level ($\epsilon$) and the accuracy probability $(1-b)$, and directly on $\sqrt{N}$. 

The rest of the paper is organized as follows. In Section~\ref{sec:prelims}, we introduce the synchronous private consensus problem, and then formally define differential privacy, convergence, and accuracy. In Sections~\ref{sec:pconsensus} and~\ref{sec:distributed}, we present and analyze the client-server and the distributed mechanisms for SPC. In Section~\ref{sec:related}, we compare our work with existing research papers in this area. In Section~\ref{sec:conclusion}, we summarize our results and discuss possible future directions.



\section{Preliminaries}
\label{sec:prelims}

For a natural number $N \in \naturals$, we denote the set $\{1, \ldots, N\}$ by $[N]$. 
For an $S$-valued vector $\theta$ of length $N$, and $i \in [N]$, we denote the $i^{th}$ component by $\theta_i$.

The mechanisms presented in this paper rely on random real numbers drawn according to the Laplace distribution.
$Lap(b)$ denotes the Laplace distribution with probability density given by $p_L(x|b) = \frac{1}{2b}e^{-|x|/b}$.
This distribution has mean 0 and variance $2b^2$. 
For any $x,y\in \reals$, $\frac{p_L(x|b)}{p_L(y|b)} \leq e^{\frac{|x-y|}{b}}$.

\subsection{Problem Statement}
\label{sec:problem}

We state the  {\em synchronous private consensus (SPC)\/} problem in the following setting. The system consists of $N$ {\em clients\/} with private initial values 
$\theta_1(0), \ldots, \theta_N(0)$ and one {\em server\/}. The clients and the server may have internal states and they communicate over channels. In each round, there are four phases: First, the clients send some messages to the server; next, the server performs computations to update its state; then it responds to the clients with some messages, and finally, the clients smoothly update their own internal states based on the response from the server.

Several vulnerabilities threaten to compromise the private initial values of the clients: (1) An intruder can have full access to all the communication channels. That is, he can peek inside all the messages going back and forth between he clients and the server. Furthermore, (2) the intruder can access the server's internal state. 

Roughly, a randomized mechanism for the clients and the server {\em solves\/} the synchronous private consensus problem if eventually all the clients converge to the average of their initial values with high probability and it guarantees that the intruder cannot learn about the initial private client values with any high level of confidence. We proceed to precisely define accuracy, convergence, and privacy.

Our definition of privacy is a modification of the notion of {\em differential privacy\/} introduced in~\cite{Dwork10} in the context of streaming algorithms.
Let $\Theta \subseteq \reals$ be the domain of individual internal states and messages. 
\begin{definition}[Adjacency]
\label{def:adj}
Two vectors $\theta, \theta'\in \Theta^N$ are $\delta$-{\em adjacent}, for some $\delta \geq 0$,  
if there exists one $i \in [N]$, such that $|\theta_i - \theta'_i| \leq \delta$  
and for all $j \neq i$, $\theta_j = \theta'_j$. 
\end{definition}
\begin{definition}[Differential Privacy]
\label{def:pri}
Let $\Theta^N \subset \reals^N$ be the domain of global state equipped with metric $m(\cdot,\cdot)$. 
Let $X$ be the set of all possible message sequences and $Y$ be the set of all possible sequences of internal states of \Alg. 
A randomized mechanism preserves {\em $\epsilon$-differential privacy\/}  if for all sets $X' \subseteq X$ and $Y' \subseteq Y$, and 
for all pairs of $\delta$-{\em adjacent\/} initial global states $\theta, \theta'\in \Theta^N$ 
\[
Pr[{\bf Alg}(\theta) \in (X',Y')] \leq e^{\epsilon \delta} Pr[{\bf Alg}(\theta') \in (X', Y')].
\]
\end{definition}

We use the standard mean square notion of convergence which has been used in the context of consensus protocols~\cite{Minyi10}.
Let $\theta_i(t)\in \reals$ be the local states of agent $A_i$ at the beginning of round $t$.
$\theta_i(0)$ denotes the secrete initial state of $A_i$.

\begin{definition}[Convergence]
\label{def:con}
A randomized mechanism is said to converge if for any initial configuration, for any $i, j\in [N]$,
 $\lim_{t\rightarrow \infty}E[(\theta_i(t)-\theta_j(t))^2]=0$, where the expectation is over the coin-flips of the algorithm.
\end{definition}




\begin{definition}[Accuracy]
\label{def:acc}
For any initial state $\theta(0)$, $b \in [0,1]$ and $r \in \nnreals$ a randomized mechanism is said to achieve $(b,r)$-{\em accuracy\/} 
if every execution starting from $\theta(0)$ converges to a state within $r$ of $\frac{1}{N} \sum_i \theta_i(0)$, with probability at least $1-b$.
\end{definition}
%
Our goal is to design a solution to the SPC problem that guaranteed to be converge.
In addition, for an adversary, looking at all a sequence of messages passing through the channels as well as a sequence of internal states of the server (and possibly some of the clients), the probability of executions corresponding to adjacent initial local states and these sequences have to be related by the Equation in Definition~\ref{def:con}. 
\section{A Client-Server Mechanism and its Analysis}
\label{sec:pconsensus}

In this section, we present a randomized mechanism for solving the synchronous private consensus problem. This mechanism has three parameters $\sigma \in (0,1)$, $c$ and $q \in (0,1)$.
 The mechanism is specified by the following client and server actions which define the four phases of each round. Let $\T = \{0\} \cup \naturals$ be the infinite time domain.
At each round $t\in \T$:

\begin{enumerate}[(i)]
\item Client $i$ sends a message $x_i(t)=\theta_i(t) + \eta_i(t)$ to the server, where $\eta_i(t)$ is a random noise generated  from the distribution $Lap(cq^t)$. 
\item The server updates its own state as the average of all client messages $y(t)= \frac{1}{N} \sum_i x_i(t)$.
\item The server sends $y(t)$ to all clients. 
\item Client $i$ updates its state by linearly interpolating between $\theta_i(t)$ and $y(t)$ with coefficient $\sigma$, that is, 
\begin{equation}
\label{eq:upd}
\theta_i(t+1) = (1-\sigma)\theta_i(t)+ \sigma y(t).
\end{equation}
\end{enumerate}

\subsection{Analysis}
\label{sec:analysis}
For $t\in \T$, let $\theta(t) = [\theta_1(t),\dots,\theta_N(t)]^T$ be the vector defining the state of the clients at the beginning of round $t$. 
Similarly, $\eta(t)$ and $x(t)$ are vectors for noise and messages. 
An {\em execution\/} of the mechanism  is an infinite sequence of the form $\alpha = \theta(0), (\eta(0), x(0), y(0)), \theta(1), (\eta(1), x(1), y(1)), \ldots$. 
Observe that given a initial vector $\theta(0)$ and the sequence of noise vectors $\eta(0), \eta(1), \ldots$, the execution of the system is completely specified. That is, for all $t \in \T$, it defines the messages $x(t), y(t)$,  the internal states of the clients $\theta(t)$ and that of the server $y(t)$.  
Thus, for brevity we will sometimes write an execution $\alpha$ as an infinite sequence of the form $\theta(0), \eta(0), \theta(1), \eta(1), \ldots$. The prefix of $\alpha$ upto round $T\in \T$ is denoted by $\alpha_T$.
We denote the set of possible executions from $\theta(0)$ as $\exec{\theta(0)}$.

For a given execution $\alpha$, the adversary can observe the subsequence of messages $x(t), y(t)$ and the server's state $y(t)$. 
We denote this subsequence by $\alpha \downarrow {(x,y)}$. 
Hence, two executions $\alpha$ and $\alpha'$ are indistinguishable to an adversary if 
$\alpha \downarrow {(x,y)} = \alpha' \downarrow {(x,y)}$.
For a set of observation sequnces $Obs$, the set of all possible executions from $\theta(0)$ which correspond to some observation in $Obs$ is the set $\exec{\theta(0),Obs} \deq \{\alpha \in \exec{\theta(0)} | \alpha \downarrow (X,Y) \in Obs\}$. 
We restate the definition of differential privacy in this context. 

\begin{definition}[Differential Privacy]
\label{def:pri}
A randomized mechanism preserves {\em $\epsilon$-differential privacy\/}  if for any set of observation sequnces $Obs$, 
and any  pairs of $\delta$-adjacent initial global states $\theta(0), \theta'(0) \in \Theta^N$ 
\begin{equation}
\label{eq:pri}
Pr[\exec{\theta(0),Obs}] \leq e^{\epsilon \delta} Pr[\exec{\theta'(0),Obs}].
\end{equation}
\end{definition}


\begin{lemma}[Privacy]
\label{lem:pri}
For $q \in (1-\sigma,1)$, the mechanism guarantees $\epsilon$-differential privacy
with $\epsilon = \frac{q}{c(q+\sigma -1)}$.
\end{lemma}

\begin{proof}
Let $\theta(0)$ and $\theta'(0)$ be arbitrary $\delta$-adjacent initial global states.
Without loss of generality, we assume that for some $k \in [N]$, $\theta_k(0) = \theta'_k(0) + \delta$.
Fix any subset of observation sequences $Obs$. We will show that Equation~\eqref{eq:pri} holds by 
establishing a bijective correspondence between the executions in $\exec{\theta(0),Obs}$ and $\exec{\theta'(0),Obs}$. 
For brevity, we denote these sets by $A$ and $A'$.

First, we define a bijection $f: A \mapsto A'$. 
For $\alpha \in A$ defined by the sequence $\eta(0), \eta(1), \ldots$, 
we define $f(\alpha) \deq \theta'(0), (\eta'(0), x'(0), y'(0)), \theta'(1), (\eta'(1), x'(1), y'(1)), \theta'(2), \ldots$,\\  
where for each 
$t \in \T$, 
\begin{eqnarray*}
\eta'_i(t) &=&
 \left\{
	\begin{array}{ll}
		\eta_i(t) + \delta (1- \sigma)^t & \mbox{for $i=k$},\\
		\eta_i(t) 											 & \mbox{otherwise}.
		\end{array}
		\right.
\end{eqnarray*}		
$x'(t) = \theta'(t) + \eta'(t)$, $y'(t) = \frac{1}{N}\sum_{i \in [N]} x'(t)$, and for $t > 0$
$\theta'(t) = (1 - \sigma) \theta'(t-1) + \sigma y'(t)$. Clearly, $f(\alpha)$ is a valid execution of the mechanism staring from $\theta'(0)$. 

The following proposition relates the states and the observable vectors of two corresponding executions. 
\begin{prop}
For all $t \in \T$, $i \in [N]$, 
\begin{enumerate}[(i)]
\item $\theta_k(t) - \theta'_k(t) = \delta (1-\sigma)^t$,
\item $\theta_i(t) = \theta'_i(t),\forall i \neq k$
\item $x'_i(t) = x_i(t)$,
\item $y'(t) = y(t)$.
\end{enumerate}
\end{prop} 
\begin{proof}
The proof is by induction on $t$. For the base case $t=0$, observe that
for $i=k$, $x'_i(0) = \theta'_i(0) + \eta'_i(0) = \theta_i(0) -\delta + \eta_i(0) + \delta =
	 x_i(0)$, otherwise, $x'_i(0) = \theta'_i(0) + \eta'_i(0) = \theta_i(0) + \eta_i(0) = x_i(0)$; 

For the inductive step, assume that the proposition holds for all $t\leq T$.
From Equation~\ref{eq:upd}, we have $\theta'_k(T+1) = (1-\sigma)\theta'_k(T)+ \sigma y'(T)$ and
$\theta_k(T+1) = (1-\sigma)\theta_k(T)+ \sigma y(T)$. The difference of these two equation gives
$\theta'_k(T+1) - \theta_k(T+1)$ 
\begin{equation*}
\begin{array}{ll}
 &= (1-\sigma)(\theta'_k(T) - \theta_k(T)) + \sigma (y'(T)-y(T))\\
 & = (1-\sigma)(\theta'_k(T) - \theta_k(T)) =\delta (1-\sigma)^{T+1}.
\end{array}
\end{equation*}
For any other client $i \neq k$, immediately from that $y'(T) = y(T)$ and $\theta'_i(T) = \theta_i(T)$, we have $\theta_i(T+1) = \theta_i(T+1)$.

Now we consider the clients' reports $x(T+1)$.
For the $k^{th}$ client, $x'_k(T+1) = \theta'_k(T+1) + \eta'_k(T+1) = \theta_k(T+1) -\delta(1-\sigma)^{T+1} + \eta_k(T+1) + \delta(1-\sigma)^{T+1} = x_k(T+1)$.
For the other client $i \neq k$, $x'_i(T+1) = \theta'_i(T+1) + \eta'_i(T+1) = \theta_i(T+1) + \eta_i(T+1) = x_i(T+1)$.
So the reports $x'(T+1) = x(T+1)$. The match up of the server's internal state immediately follows.
\end{proof}
Parts $(iii)$ and $(iv)$ of the above proposition establishes that $\alpha$ and $f(\alpha)$ are indistinguishable, that is, indeed they produce the same observation sequence.

Next we will relate the probability of any finite prefix of an {\em individual\/} execution $\alpha \in A$, 
and its corresponding execution $f(\alpha) \in A'$, for a particular observation sequence $\beta \in Obs$: 
\begin{eqnarray*}
&&\frac{Pr[\alpha_T = \theta(0), \ldots, \theta(T)]}{Pr[(f(\alpha))_T = \theta(0), \ldots, \theta(T)]} \\ 
&=& \prod_{t =0}^{T-1} \prod_{i \in [N]} \frac{p_L(\eta'_i(t)|cq^t)}{p_L(\eta_i(t)|cq^t)} 
= \prod_{t =0}^{T-1} \frac{p_L(\eta'_k(t)|cq^t)}{p_L(\eta_k(t)|cq^t)} \\
&\leq& \prod_{t =0}^{T-1} e^{\frac{|\eta'_k(t) - \eta'_k(t)|}{cq^t}} = \prod_{t =0}^{T-1} e^{\frac{\delta}{c} \left(\frac{1-\sigma}{q}\right)^t}.
\end{eqnarray*}
Integrating over all executions $\alpha \in A$, we get 
\begin{eqnarray*}
&&\int_{\alpha \in A} Pr[\alpha_T = \theta(0), \ldots, \theta(T)] d\mu \\
&\leq& \prod_{t =0}^{T-1} e^{\frac{\delta}{c} \left(\frac{1-\sigma}{q}\right)^t} 
\int_{f(\alpha) \in A'} Pr[(f(\alpha))_T = \theta'(0), \ldots, \theta'(T)] d\mu', 
\end{eqnarray*}
where $d\mu$ and $d\mu'$ are probability measures over $A$ and 
$A'$ defined by the randomized mechanism. 
If $q \in (1-\sigma,1)$, then as $T \rightarrow \infty$, 
the product  converges to $e^{\epsilon \delta}$, where $\epsilon = \frac{q}{c(q+\sigma -1)}$,
and we obtain the required inequality for $\epsilon$-differential privacy.
\begin{eqnarray*}
Pr[\exec{\theta(0), Obs}] \leq  e^{\epsilon \delta} Pr[\exec{\theta'(0), Obs}].
\end{eqnarray*}
\end{proof}

\begin{lemma}[Convergence]
\label{lem:con}
The mechanism described above achieves convergence.
\end{lemma}

\begin{proof}
We define a global potential function $P: \naturals \rightarrow \nnreals$ as 
$P(t) = \frac{1}{2}\sum_{i\neq j} [\theta_i(t)-\theta_j(t)]^2$. 
Using the matrix notation $P(t) = \theta(t)^T \ L \ \theta(t)$,
where $L\in \reals^{N \times N}$ with elements:
\begin{equation}
\label{eq:L}
l(i,j) =  \left\{
\begin{array}{ll}
N-1, & i=j, \\
-1, & \mbox{otherwise}.
\end{array}
\right.
\end{equation}
The transition rule for the internal state of the $i^{th}$ client can be written as: 
\begin{equation}
\label{eq:loc_upd}
\begin{array}{rrl}
\theta_i(t+1) &=&  (1-\sigma) \theta_i(t) + \frac{\sigma}{N} \sum_{i=1}^{N} (\theta_i(t) + \eta_i(t)) \\
&=&  (1+\frac{\sigma}{N}-\sigma) \theta_i(t) + \frac{\sigma}{N} \sum_{j\neq i}\theta_j(t)+ \frac{\sigma}{N} w(t),
\end{array}
\end{equation}
where $w(t) = \sum_i \eta_i(t)$.  
The update rule for all the agents can be written as 
$\theta(t+1) = \theta(t) - \frac{\sigma}{N} L \theta(t) + \frac{\sigma}{N}  w(t)\mathbbm{1}_N$. Then,
\begin{equation}
\begin{array}{rl}
P(t+1) &= \theta(t+1)^TL\theta(t+1) \\
 &=[\theta(t) - \frac{\sigma}{N} L \theta(t) + \frac{\sigma}{N} w(t)\mathbbm{1}_N]^TL \\
 & \quad       [\theta(t) - \frac{\sigma}{N} L \theta(t) + \frac{\sigma}{N} w(t)\mathbbm{1}_N]\\
 &= P(t) - 2\frac{\sigma}{N}\theta(t)^TLL\theta(t) + \frac{\sigma^2}{N^2} \theta(t)^TLLL\theta(t)\\
 & \quad + 2\frac{\sigma}{N}w(t)\theta(t)^TL\mathbbm{1}_N- \frac{2\sigma^2}{N^2}w(t)\theta(t)^TLL\mathbbm{1}_N  \\
 & \quad  +\frac{\sigma^2}{N^2}w(t)^2 \mathbbm{1}^T_N L\mathbbm{1}_N. \\
 & \label{eq:lya}
=  P(t) - 2\frac{\sigma}{N}\theta(t)^TLL\theta(t) + \frac{\sigma^2}{N^2} \theta(t)^TLLL\theta(t).
 \end{array}
\end{equation}
By Equation~\ref{eq:L} we have $L = N I_N - \mathbbm{1}_{NN}$. So in this particular case, we have $L L = (N I_N - \mathbbm{1}_{NN})^2 = N^2I_N-2N\mathbbm{1}_{NN} + \mathbbm{1}_{NN}^2 = N^2I_N-N\mathbbm{1}_{NN} = N L$. Similarly $LLL = N^2L$. Substitute the previous equation into Equation~\eqref{eq:lya} we get,
\[
P(t+1) = (1-2\sigma+\sigma^2)P(t) = a P(t),
\]
where $a = (1-\sigma)^2$.
For all $\sigma\in (0,1)$, $a\in(0,1)$ is a constant.
Thus we have as $t \rightarrow \infty$, $P(t)$ converges exponentially to $0$, which implies convergence.
\end{proof}
From Equation~\eqref{eq:loc_upd}, each agent adds an {\em identical} random variable $\frac{\sigma}{N}w(t)$ to its local state in round $t$.
Although the average value drifts with this random variable, the relative distance between local states will not be affected. 
As a result, the mechanism converges deterministically.  

\begin{lemma}[Accuracy]
\label{lem:acc}
For any $b \in (0,1)$, the randomized mechanism achieves $(b,\frac{\sqrt{2}c\sigma}{\sqrt{b N(1-q^2)}}))$-accuracy.
\end{lemma}

\begin{proof}
This is a special case of a more general proof we show later.
Please see the proof of Lemma~\ref{lem:acc_dis} with $\tilde{d}=\frac{\sigma^2}{N}$ particularly for this case.
%
\end{proof}

In this section we proposed an solution to the centralized synchronous consensus problem and 
formally established its privacy, convergence and accuracy properties. We will discuss the trade-offs between privacy and accuracy in Section~\ref{sec:discussion}.




\section{A Distributed Mechanism}
\label{sec:distributed}

In this section, we present a second synchronous randomized mechanism for solving the private consensus problem which does not use a server but instead relies on the clients exchanging information with their neighbors in a truly distributed fashion. Let $G = ([N],\E)$ be a {\em undirected connected graph}, where $[N]$ is the set of {\em vertices \/} and $\E\subset [N] \times [N]$ is the set of {\em edges\/}. Let $N(i) = \{j\in [N] | (i,j) \in \E \}$ be the set of {\em neighbors\/} of node $i$ with whom it communicates. Let $|N(i)|$ be the {\em degree\/} of node $i$ in $G$. 

As in the previous setting, an intruder has access to all the communication channels as well as the internal states of a set $C$ of compromised clients (but cannot overwrite them). Our mechanism will protect the privacy of clients who are not compromised. Thus, in this context, Definition~\ref{def:pri} is modified by restricting the notion of $\delta$-adjacency to uncompromised agents.

Now we state a mechanism to solve the distributed SPC problem. 
Besides the state variable $\theta_i$ which holds the consensus value, client $i$ holds another auxiliary state $y_i$.
The mechanism has parameters $\sigma \in (0,1)^N$, $c$ and $q\in(0,1)$. 
Instead of sharing an identical linear combination factor, client $i$ has an independent $\sigma_i \in (0,1)$ which is the $i^{th}$ element of vector $\sigma$. 
At each round $t \geq 0$:
\begin{enumerate}[(i)]
\item Client $i$ sends a message $x_i(t)=\theta_i(t) + \eta_i(t)$ to every $j \in N(i)$, where $\eta_i(t)$ is a random noise generated from the distribution $Lap(cq^t)$. 
\item Client $i$ updates $y_i$ as the average of $x_i(t)$ and the messages it receives: 
\begin{equation}
\label{eq:y_dis}
y_i(t)= \frac{1}{|N(i)|+1} \sum_{j\in N(i) \cup \{i\}} x_j(t).
\end{equation}
\item Client $i$ updates $\theta_i$  by linearly interpolating between $\theta_i(t)$ and $y_i(t)$ with coefficient $\sigma_i$, that is, 
\begin{equation}
\label{eq:upd_dis}
\theta_i(t+1) = (1-\sigma_i)\theta_i(t)+ \sigma_i y_i(t).
\end{equation}
\end{enumerate}

\subsection{Analysis}
\label{sec:ana_dis}

The analysis of the  distributed mechanism parallels the analysis presented in Section~\ref{sec:pconsensus}. 
An execution $\alpha$ is defined similar to the centralized setting except that $y(t)$ in this case is a vector rather than a scaler.
The privacy of those corrupted nodes makes no sense. Let $C \subset \N$ be the set of corrupted nodes.

\begin{lemma}[Privacy]
\label{lem:pri_dis}
For $q \in (1-\sigma_m,1)$, where $\sigma_m$ is the minimum element of vector $\sigma$, the distributed mechanism guarantees $\epsilon$-differential privacy with respect to the uncorrupted nodes with $\epsilon = \frac{q}{c(q+\sigma_m-1)}$.
\end{lemma}
We omit the proof of Lemma~\ref{lem:pri_dis} as it is a straight forward generalization of the proof of Lemma~\ref{lem:pri}.

In contrast to Lemma~\ref{lem:con}, the convergence of the distributed mechanism depends on the structure of graph $G$. 
Before stating the convergence result, we introduce Laplacian matrix $L$ of graph $G$ with elements:
\begin{equation}
\label{eq:L_dis}
l(i,j) =  \left\{
\begin{array}{ll}
|N(i)| & i=j, \\
-1 & (i,j)\in \E, \\
0 & \mbox{otherwise}.
\end{array}
\right.
\end{equation}
The Laplacian matrix $L$ for any graph is known to have several nice properties. It is by definition symmetric with real entries, hence it can be diagonalized by an orthogonal matrix.
It is positive semidefinite, hence its real eigenvalues can be ordered as $\lambda_1 \leq \lambda_2\leq \ldots \leq \lambda_N$  be the eigenvalues of $L$. Furthermore $\lambda_1 = 0$ and $\lambda_2 > 0$ if and only if the graph is connected. 
Let $\{v_1,v_2,\ldots,v_N\}$ be a set of orthonormal eigenvectors of $L$ such that $v_k$ corresponds to $\lambda_k$. In addition, denote $d_i = \frac{\sigma_i}{|N(i)|+1}$.
We state a sufficient condition of convergence as following.
\begin{assumption}
\label{ass:con}
Assume that graph $G$ has the following properties.
\begin{enumerate}[(I)]
\item $\lambda_2>0$, that is graph $G$ is connected. 
\item $\lambda_N<\frac{M^2}{2m}$, where $m = \inf_{i\in [N]}d_i$ and $M = \sup_{i\in [N]}d_i$.
\end{enumerate}
\end{assumption}
\begin{lemma}[Convergence]
\label{lem:con_dis}
The distributed mechanism described above achieves convergence if Assumption~\ref{ass:con} holds. 
\end{lemma}
\begin{proof}
We define a function $P:\naturals \mapsto \nnreals$ as 
\[
P(t)=\frac{1}{2}\sum_{(i,j)\in \E}[\theta_i(t)-\theta_j(t)]^2.
\] 
Using the matrix notation $P(t) = \theta(t)^T \ L \ \theta(t)$.
By Assumption~\ref{ass:con}, $E[P(t)] = 0 \ \Leftrightarrow \sum_{i\neq j}E[\theta_i(t)-\theta_j(t)]^2 = 0$.
According to Equation~\eqref{eq:y_dis} and~\eqref{eq:upd_dis}, the update equation of client $i$ is:
\begin{equation}
\label{eq:loc_upd_dis}
\begin{array}{rl}
\theta_i(t+1) =& (1-d_i |N(i)|)\theta_i(t) 
  + d_i\sum_{j\in N(i)}\theta_j(t) \\
	& \  + d_i w_i(t),
\end{array}
\end{equation}
where 
\begin{equation}
\label{eq:w_def}
w_i(t)=\sum_{j\in N(i)\cup\{i\}}\eta_i(t).
\end{equation}
We define vector $w(t)=[w_1(t),\ldots,w_N(t)]^T$ and matrix $D\in \reals^{N \times N}$ with elements:
\begin{equation}
\label{eq:D}
d(i,j) =  \left\{
\begin{array}{ll}
d_i, & i=j, \\
0, & \mbox{otherwise}.
\end{array}
\right.
\end{equation}
The update rule for all the agents can be written as 
$\theta(t+1) = \theta(t) -  D L \theta(t) +  D  w(t)$.
Then, $P(t+1)$
\begin{equation}
\label{eq:P}
\begin{array}{ll}
=& \theta(t+1)^TL\theta(t+1) \\
=&(\theta(t) -  D L \theta(t) +  D  w(t))^T  L  (\theta(t) -  D L \theta(t) +  D  w(t))\\
=&P(t)-2\theta(t)^TLDL\theta(t) + \theta(t)^TLDLDL\theta(t) + \\
& \quad 2\theta(t)^T(I-DL)LD w(t) + w(t)^T DLD w(t).
\end{array}
\end{equation}
Taking expectation of both sides with respect to the coin flips of the algorithm starting from any state:  
\begin{equation}
\label{eq:P_exp}
\begin{array}{rl}
E[P(t+1)] &= E[P(t)] - E[Q(\theta(t))] + E[w(t)^TDLDw(t)],
\end{array}
\end{equation}
where, 
\[Q(\theta) = 2\theta^TLDL\theta - \theta^TLDLDL\theta.\]
The term $E[2\theta(t)^T(I-DL)LD w(t)]$ vanishes because
\begin{inparaenum}[(i)]
\item $\theta(t)$ and $w(t)$ are independent; and 
\item by Equation~\eqref{eq:w_def}, $w(t)$ has zero mean.
\end{inparaenum}

Now we will prove that there exists a constant $a \in (0,1)$ such that $Q(\theta(t)) \geq a P(t)$.
Because $L$ is positive semidefinite, we have $0\leq L \leq \lambda_N I$.
From Assumption~\ref{ass:con} and Equation~\eqref{eq:D}, we have $mI \leq D \leq MI$. 
Then,
\begin{equation}
\label{eq:Qbnd}
\begin{array}{rl}
Q(\theta) &\geq 2m\theta^TLL\theta - \lambda_N\theta^TLDDL\theta  \\
& \geq 2m \theta^TLL\theta - \lambda_NM^2\theta^TLL\theta \\
& \geq (2m - \lambda_NM^2)\theta^TLL\theta.
\end{array}
\end{equation}
%
%
%
The following proposition helps obtain a bound on  $a$.
\begin{prop}
\label{prop:bnd}
For any  $\theta\in \reals^N$, $\theta^T LL \theta \geq \lambda_2 \theta^T L \theta$.
\end{prop}
\begin{proof}
First, we show that the proposition holds for any eigenvector $v_k$ of $L$.
For the eigenvector $v_1$ corresponding to $\lambda_1=0$, we have $v_1^T L = 0$ and the inequality holds trivially. 
For any other eigenvector $v_k$ and the corresponding eigenvalue $\lambda_k>0$, we have $v_k^T LL v_k = \lambda_k v_k^T L v_k \geq \lambda_2 v_k^T L v_k$. 
Next, we  prove that the proposition holds for any vector $\theta$.
Because $\{v_1,v_2,\ldots,v_N\}$ is an orthonormal basis, 
for any $i\neq j$, $v_i^TLLv_j = \lambda_j v_i^TLv_j =\lambda_j^2v_i^Tv_j = 0$. 
For any $\theta = \sum_{k\in[N]} \alpha_kv_k$,
we have: 
\[ \begin{array}{rl}
\theta^T LL \theta &= (\sum_{k\in[N]} \alpha_kv_k)^T LL (\sum_{k\in[N]} \alpha_kv_k) \\ 
& = \sum_{k\in[N]} \alpha_k^2v_k^TLLv_k \\
& \geq \lambda_2 \sum_{k\in[N]} \alpha_k^2v_k^TLv_k  = \lambda_2 \theta^T L \theta.
\end{array} \]
\end{proof}
From Equation~(\ref{eq:Qbnd}), then it follows that
\[
Q(\theta(t)) \geq \lambda_2 (2m - \lambda_N M^2) P(t).
\]
Thus, for any $ a \leq \min(\lambda_2 (2m - \lambda_N M^2), 1)$, the inequality $Q(\theta(t)) \geq a P(t)$ holds.
%
Also, by Assumption~\ref{ass:con}, $\lambda_2(2m - \lambda_NM^2)>0$.
Then, for some $a\in (0,1)$, Equation~\eqref{eq:P_exp} is reduced to 
\[ \begin{array}{rl}
E[P(t+1)] &\leq(1- a)E[P(t)] + E[w(t)^TDLDw(t)]\\
&\leq (1-a)E[P(t)] + \lambda_NM^2E[w(t)^Tw(t)].
\end{array}
\]
As $t \rightarrow \infty$ the contribution of the first term converges to $0$. 
For the second term, recall that each element of $w(t)$ is a linear combination of i.i.d $\eta_i(t) \sim Lap(cq^t)$. 
For $i \neq j$, $E[\eta_i(t)\eta_j(t)]=E[\eta_i(t)]E[\eta_j(t)]=0$. 
For any $i$, $E[\eta_i(t)^2]=Var(\eta_i(t)) = 2c^2q^{2t}$, which also converges  to 0.
So $E[w(t)^Tw(t)] \rightarrow 0$ as $t \rightarrow \infty$. 
Combining, we have $E[P(t)] \rightarrow 0$ as $t \rightarrow \infty$. 
\end{proof}
In general, the expected consensus value of the distributed algorithm does not coincide with the initial average. 
Intuitively, a node with higher degree or slower evolution will have heavier weight on the consensus value.
In this context, Definition~\ref{def:acc} is modified by replacing the average $\bar{\theta}(0)=\frac{1}{N} \sum_i \theta_i(0)$ with a weighted modification $\bar{\theta}(0) = \frac{\sum_i\gamma_i\theta_i(0)}{\sum_i\gamma_i}$, where the weight $\gamma_i =\frac{1}{d_i} = \frac{|N(i)|+1}{\sigma_i}$. 
\begin{lemma}[accuracy]
\label{lem:acc_dis}
The distributed mechanism achieves $(b,\frac{\sqrt{2\tilde{d}} c}{\sqrt{b (1-q^2)}})$-accuracy, where $\tilde{d} =  \frac{\sum_i(|N(i)|+1)^2}{(\sum_i\gamma_i)^2}$.
\end{lemma}
\begin{proof}
Let us fix an initial state $\theta(0)$ and define $\bar{\theta}(t) = \frac{\sum_i\gamma_i\theta_i(t)}{\sum_i\gamma_i}$ and $\tilde{w}(t) = \frac{\sum_iw_i}{\sum_i\gamma_i}$.
We rewrite Equation~\eqref{eq:loc_upd_dis} with 
\[
\gamma_i\theta_i(t+1) = \gamma_i\theta_i(t) - |N(i)|\theta_i(t) + \sum_{j\in N(i)}\theta_j(t) + w_i(t).
\]
Add up all $N$ equations and divided by $\sum_i\gamma_i$, we get:
\[
\bar{\theta}(t+1) = \bar{\theta}(t) +  \tilde{w}(t)= \bar{\theta}(0) + \sum_{s=0}^{t}\tilde{w}(s).
\]
From the definition of $\tilde{w}(t)$ and Equation~\eqref{eq:w_def}, we have 
\begin{eqnarray*}
Var(\tilde{w}(t)) &=&  \frac{Var(\sum_iw_i(t))}{(\sum_i\gamma_i)^2}= \frac{Var(\sum_i(|N(i)|+1)\eta_i(t))}{(\sum_i\gamma_i)^2} \\
&=& \frac{Var(\eta_i(t))\sum_i(|N(i)|+1)^2}{(\sum_i\gamma_i)^2}= 2\tilde{d}c^2q^{2t}.
\end{eqnarray*}
By $q\in (0,1)$, the series converges. 
\[
Var(\sum_{s=0}^t \tilde{w}(s))\leq Var(\sum_{s=0}^\infty \tilde{w}(s))=\frac{2\tilde{d}c^2}{1-q^2}.
\]
By Chebyshev's inequality for any $t\geq 0$: 
\[
Pr(|\bar{\theta}(t)-\bar{\theta}(0)| \leq r) 
= 1 - Pr(|\sum_{s=0}^t w(s)| > r) 
\geq 1-\frac{Var(\sum_{s=0}^t w(s))}{r^2}.
\]
Choosing $r = \frac{\sqrt{Var(\sum_{s=0}^t w(s))}}{\sqrt{b}} = \frac{\sqrt{2\tilde{d}}c}{\sqrt{b N(1-q^2)}}$, we have 
$1 - Pr(|\sum_{s=0}^t w(s)| > r) \geq 1-b$.
Let $t\rightarrow \infty$, by Lemma~\ref{lem:con_dis} every execution converges. Then the lemma follows.
\end{proof}

The trade-off between accuracy and privacy of this mechanism is similar to that of the client-server mechanism of Section~\ref{sec:pconsensus} and we discuss them together next.

\section{Discussion on Results}
\label{sec:discussion}

We proposed two mechanisms that achieve iterative private consensus over infinite horizon by adding a stream of noises to the messages set by the clients (to each other or to the server). 
The standard deviation of the Laplace distribution of the noise added in every round decreases and ultimately converges  to $Lap(0)$ which is the Dirac $\delta$ distribution at $0$. 
The mechanisms have 3 parameters: linear combination factor $\sigma$, initial noise $c$ and noise convergence rate $q$. 
The constraint to achieve privacy over infinite horizon is that $q>1-\sigma$, 
which roughly means that the noise should converge slower than the system's inertia 
so as to ``cover'' the trail of dynamics.

From Lemma~\ref{lem:pri} and \ref{lem:pri_dis} we observe that $\epsilon$ decreases with larger $c$ or $q$.
This implies that the system has a higher privacy if the noise values are picked from a Laplace distribution with larger parameters (and hence larger standard deviation). 
From Lemma~\ref{lem:acc} and \ref{lem:acc_dis}, however, a more dispersive noise results in worse accuracy. 
The tradeoff between privacy and accuracy for different noise convergent rate ($q$) is illustrated in Figure~\ref{fig:tradeoff}.
If we fix the parameter $q$, we observe that for $\epsilon$-differential privacy for $N$ agents and an accuracy level of $b$ the accuracy radius $r$ is $O(\frac{1}{\epsilon\sqrt{\beta N}})$. For specific values on these parameters, the dependence between $\epsilon$ and $r$ is shown in Figure~\ref{fig:tradeoff}.
\begin{figure}[bth]
\centering
\caption{Privacy and Accuracy as functions of the Noise convergent rate in the centralized mechanism.
	Parameterized with $N = 500$, $\sigma = 0.8$, $c = 10$ and $\beta = 0.5$.}
	\includegraphics[width=0.48\textwidth]{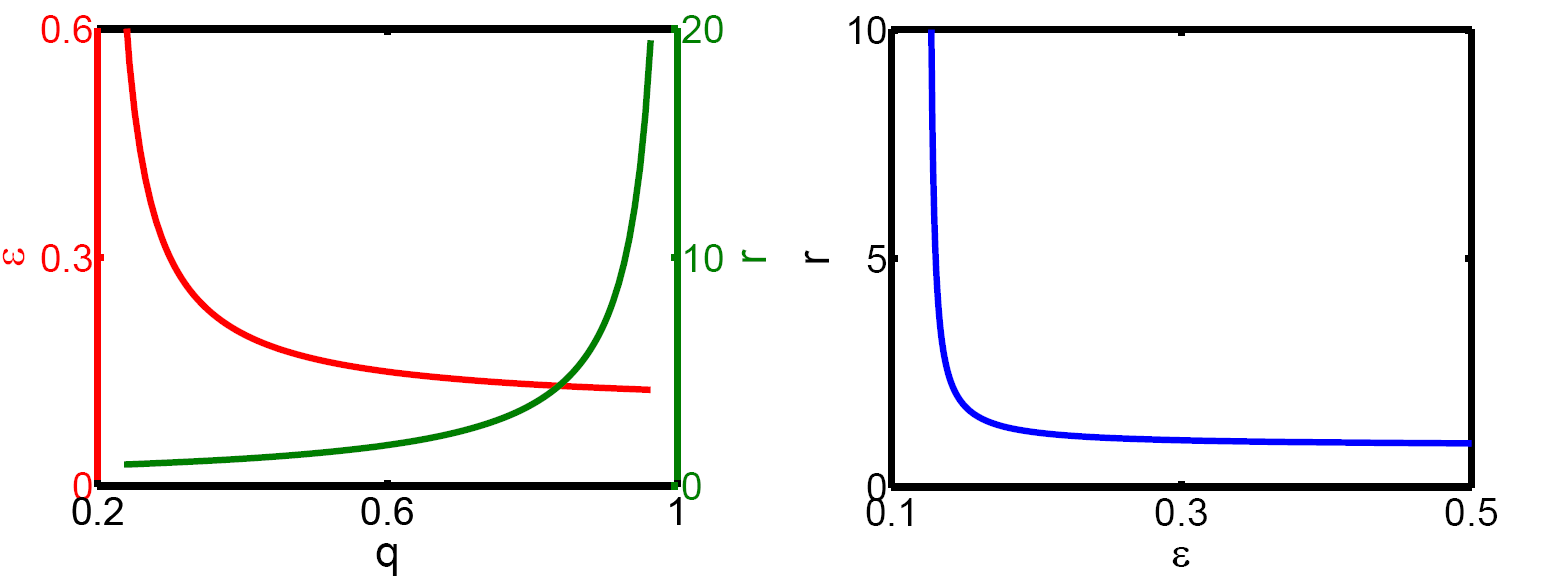}
	\label{fig:tradeoff}
\end{figure}

\section{Related Work}
\label{sec:related}

Our consensus mechanism  has similarities with  the protocols for computing sum and inner product presented in~\cite{DBLP:journals/corr/abs-1111-5228}, in that, all these protocols rely on adding noise to the states communicated among the participants.  Our mechanism differs in the type of noise (geometrically decaying Laplace) that is added. Moreover, in our setup, the computed outputs are used as feedback for updating the state of the participants to achieve convergence. 

In~\cite{Duan:2010:PPL:1929820.1929839} a framework for securely computing general types of aggregates  is presented. Every client splits its private data into pieces and sends them to different servers. If at least one server is not compromised, then the iterative aggregate computation is guaranteed to preserve privacy of the individuals. Our mechanism is quite different 
and it guarantees privacy even if the only server is compromised.

In~\cite{Xiong:2007:PDP:1275505.1275510}, the authors present distributed protocols for computing  $k$ maximum values among all participants. In this protocol, the clients communicate a global vector of $k$-maximum values over a ring network. In each step, the client processing the global vector either with an exponential decaying probability honestly replaces the values in global state if it is smaller than one of the local values, 
or it replaces the values in the vector with randomly generated small numbers. 
The metric of privacy is {\em Loss of Privacy\/} which characterizes the additional knowledge to the adversary of gaining intermediate result besides the final results. 
This work is setup with quite different definition of privacy compare to ours. In addition, some features of our mechanism, such as feedback update and infinite horizon, are not presented in this protocol.

\section{Conclusions and Future direction}
\label{sec:conclusion}

In this paper, we formalize a Synchronized Private Consensus problem and propose two mechanisms for solving it. The first one relies on the client-server model of communication and the latter is purely distributed.  The key idea is to add a random noise to the clients' messages to the server (or other clients) that is drawn from a Laplace distribution that converges to the Dirac distribution. The messages with large noice give differential privacy and as the noise level attenuates, the system converges to the target value with probability that depends inversely on the security parameter and directly in the number of participants. 
The feedback $y(t)$ from the server is the mean of all noisy messages sent. And, the clients update their states by taking a linear combination of $y(t)$ and their previous state. We formally prove the privacy and convergence of this mechanism. The key proof technique for privacy,  relies on constructing a bijective map between two sets of executions starting from different but adjacent initial states. 


To the best of our knowledge this is the first investigation of differential privacy in the context of control systems where the ultimate goal is convergence. Our results suggest several directions for future work.
First, we are trying to apply our method to a larger set of control problems that arise from iterative closed-loop control.
Novel applications of this arise from  differential privacy and more generally security of  {\em distributed cyber-physical systems\/} where the  physical state is updated smoothly according to some differential equations.

Second, we also interested in exploring the tradeoff between privacy and performance under more general dynamics of the system. In the SPC problem we discussed, the dynamics of the system is discrete and linear. We expect to extend the analysis to continuous or non-linear systems. Also, establishing a lower bound for the problem will be of significance.

An orthogonal direction is to develop automated verification and synthesis algorithms for controllers that preserve differential privacy. Along these lines,  a verification framework for streaming algorithms has been presented in~\cite{Barthe:reason.privacy,Kaynar:verification.pri}. The challenge will be to extend these ideas to synthesis and feedback control systems.

\bibliographystyle{abbrv}
\bibliography{sayan1,Privacy}

\end{document}